\documentclass{article}
\usepackage[utf8]{inputenc}

\title{MacWilliams' Extension Theorem for rank-metric codes}
\author{Elisa Gorla and Flavio Salizzoni}
\date{}

\usepackage{amssymb,amsmath,tikz,amsthm}
\usepackage{enumerate,hyperref}

\usepackage{geometry}
\theoremstyle{definition}

\newtheorem{lemma}{Lemma}[section]
\newtheorem{proposition}[lemma]{Proposition}
\newtheorem{theorem}[lemma]{Theorem}
\newtheorem*{maintheorem}{Main Theorem}
\newtheorem{corollary}[lemma]{Corollary}

\newtheorem{question}[lemma]{Question}
\newtheorem{definition}[lemma]{Definition}
\newtheorem{remark}[lemma]{Remark}
\newtheorem{example}[lemma]{Example}

\newtheorem*{ext}{Extension Property}
\newtheorem{property}{Property}
\newtheorem*{mwthm}{MacWilliams' Extension Theorem}

\newcommand{\Cc}{\mathcal{C}}
\newcommand{\F}{\mathbb{F}}
\newcommand{\Id}{\mathrm{Id}}
\newcommand{\rk}{\mathrm{rank}}
\newcommand{\rs}{\mathrm{rowsp}}
\newcommand{\cs}{\mathrm{colsp}}

\DeclareMathOperator{\rowsp}{rowsp}
\DeclareMathOperator{\colsp}{colsp}
\DeclareMathOperator{\gl}{GL}

\begin{document}

\maketitle

\begin{abstract}
The MacWilliams' Extension Theorem is a classical result by Florence Jessie MacWilliams. It shows that every linear isometry between linear block-codes endowed with the Hamming distance can be extended to a linear isometry of the ambient space. Such an extension fails to exist in general for rank-metric codes, that is, one can easily find examples of linear isometries between rank-metric codes which cannot be extended to linear isometries of the ambient space. In this paper, we explore to what extent a MacWilliams' Extension Theorem may hold for rank-metric codes. We provide an extensive list of examples of obstructions to the existence of an extension, as well as a positive result.
\end{abstract}

\section*{Introduction and motivation}

Coding theory provides tools for the transmission and storage of data over an imperfect channel, where the data may be altered or lost. One of the main goals is being able to automatically correct errors in a received message, without asking for a retransmission. This is done through the use of (error-correcting) codes: The data to be sent is encoded, i.e., transformed into a codeword by adding redundancy to it. The set of codewords is called a code. The codeword travels over the channel, where part of the information may be lost or corrupted. At the receiver's end, the received information is decoded, that is, the error is corrected and the redundancy eliminated. In the mathematical formulation of error-correcting codes, we usually ignore the step in which the redundancy is eliminated, since it does not present any theoretical or practical challenges.

In many scenarios, error correction is done via minimum distance decoding. A code is a subset of a finite metric space and a received message is decoded to the closest codeword. Mathematically, if $(S,d)$ is a finite metric space and $C\subseteq S$ a code, a received $r\in S$ is decoded to an $x\in C$ which minimizes $d(-,r)$. Under suitable assumptions, the $x$ which minimizes $d(-,r)$ is unique. One way to guarantee uniqueness is as follows: Define the minimum distance of a code $C$ as $$d_{\min}(C)=\min\{d(x,y)\mid x,y\in C, x\neq y\}.$$ It is easy to show that, given $r\in S$, if there is an $x\in C$ such that $d(x,r)<(d_{\min}(C)-1)/2$, then $x$ is the unique codeword which minimizes $d(-,r)$. The quantity $(d_{\min(C)}-1)/2$ is often called the error-correction capability of the code. 

This motivates the interest for isometries between codes, since these are the maps that preserve the pairwise distances of codewords, therefore the metric structure of the code, and in particular its error-correction capability. However, one could also look at isometries of the ambient space $\varphi:S\rightarrow S$. Such an isometry does not only preserve the metric structure of the code, mapping $C$ to an isometric code $\varphi(C)$, but also the distance between any pair of elements of $S$, that is $d(x,r)=d(\varphi(x),\varphi(r))$ for any $x,r\in S$. In particular, $\varphi$ preserves the whole error correction procedure, in the sense that $r\in S$ is decoded to $x\in C$ if and only if $\varphi(r)\in S$ is decoded to $\varphi(x)\in\varphi(C)$. In some cases, we know that any isometry between codes is the restriction of an isometry of the ambient space $S$, that is, any isometry between codes can be extended to an isometry of the ambient space. In this paper, we call this property the Extension Property. 

Linear block codes endowed with the Hamming distance are used in point-to-point communication. These are linear subspaces of $\F_q^n$, where $\F_q$ denotes the finite field with $q$ elements. In \cite{MacW62} Florence Jessie MacWilliams showed that every Hamming distance-preserving linear isomorphism $\varphi:\Cc_1\rightarrow\Cc_2$ between two codes in $\F_q^n$ can be extended to a Hamming distance-preserving linear isomorphism $\mu:\F_q^n\rightarrow\F_q^n$. An elementary proof of this fact was later given by Kenneth Bogart, Don Goldberg and Jean Gordon in \cite{BGG78}. Nowadays, this theorem is known as the MacWilliams' Extension Theorem.

\begin{mwthm}
Every linear Hamming weight isometry $\varphi$ of linear codes over a finite field $\F_q$ extends to a linear Hamming weight isometry $\mu$ of the ambient space $\F_q^n$.
\end{mwthm}

In the last decades, there has been an increasing interest in understanding for which ambient spaces and for which weights a similar Extension Property holds. In \cite{Woo97,Woo99} Jay Wood studied the case of finite rings and established the Extension Property for codes over finite Frobenius rings with respect to the Hamming distance. Aleams Barra and Heide Gluesing-Luerssen investigated further the case of finite Frobenius rings with various distance functions in \cite{BG13}. Friedrich Martin Schneider and Jens Zumbrägel extended the work of Wood to Artinian rings in \cite{SZ19}. Recently, the Extension Property was proved in \cite{Dys19,LW19} for codes over $\mathbb{Z}_m$ endowed with the Lee distance.

In this paper, we explore the Extension Property in the setting of rank-metric codes. These are linear spaces of matrices inside $\F_q^{m\times n}$, where $\F_q$ is the finite field with $q$ elements. The rank distance between two matrices is the rank of their difference. Rank-metric codes are useful for correcting errors and increasing the efficiency of data transmission over a network. 

\begin{ext}
Let $\Cc_1,\Cc_2$ be two linear codes in $\F_q^{m\times n}$. A linear isometry $\varphi:\Cc_1\rightarrow\Cc_2$ satisfies the Extension Property if and only if there exists a linear isometry $\mu:\F_q^{m\times n}\rightarrow\F_q^{m\times n}$ such that $\mu |_{\Cc_1}=\varphi$.
\end{ext}

It is well known that there exist isometries of rank metric codes that do not satisfy the Extension Property (see \cite{BG13} and \cite[Section 7]{CKWWv1}). We are interested in understanding under which conditions it may be possible to extend an isometry to the whole ambient space and when instead the Extension Property fails. Very little is know in this direction. The results in \cite{GHFWZ} imply that isometries between two rank support spaces are extendable. The same result for $\F_{q^m}$-isometries between Galois closed linear subspaces of $\F_{q^m}^n$ was proved by Umberto Martínez-Peñas in \cite[Theorem 5]{Mar16}. 

In Section 1, we recall some definitions and results on rank-metric codes. In Section 2 we present an extensive list of obstructions to the Extension Property, providing multiple examples, while in Section 4 we establish the Extension Property in a special case. Section 3 is dedicated to developing some tools that are used in Section 4. Our Main Theorem states that the Extension Property holds for certain isometries of codes generated by elementary matrices. In the appendix, we establish some mathematical facts connected to the proof of the Main Theorem in Section 4.

\section{Preliminaries on rank-metric codes}

Throughout this paper, $q$ is a prime power and $\F_q$ denotes the finite field with $q$ elements. For positive integers $m,n$, we denote by $\F_q^{m \times n}$ the set of $m \times n$ matrices with entries in $\F_q$. We denote by $\rk(M)$ the rank of a matrix $M \in \F_q^{m \times n}$ and by $\dim(V)$ the dimension of an $\F_q$-linear space $V$. 

\begin{definition}
The {\bf rank distance} of $A, B \in\F_q^{m\times n}$ is defined as
\begin{align*}
 d :  \F_q^{m\times n} \times \F_q^{m\times n} & \longrightarrow  \mathbb{N} \\
      (A,B)     \qquad        & \longmapsto      \rk(A-B). \\
\end{align*}
A {\bf rank-metric code} $\mathcal{C}\subseteq\F_q^{m\times n}$ is an $\F_q$-linear subspace endowed with the rank distance. 
\end{definition}

In order to properly state the Extension Property in the context of rank-metric codes, we briefly recall the notion of isometric and equivalent codes.

\begin{definition}
Let $\Cc_1,\Cc_2$ be two linear codes in $\F_q^{m\times n}$. An $\F_q$-linear isomorphism $\varphi:\Cc_1\rightarrow\Cc_2$ such that $\rk(C)=\rk(\varphi(C))$ for all $C\in\Cc_1$ is called {\bf isometry} and $\Cc_1,\Cc_2$ are {\bf isometric}.
\end{definition}

The following classification of the linear isometries of $\F_q^{m\times n}$ is due to Hua \cite{Hua} for odd characteristic and to Wan \cite{Wan} for characteristic 2. The statement can also be found in \cite[Theorem 11.1.9]{Gor21}.

\begin{theorem}\label{classification}
Let $\varphi:\F_q^{m\times n}\rightarrow \F_q^{m\times n}$ be an $\F_q$-linear isometry with respect to the rank metric.
\begin{enumerate}[(a)]
\item If $m\neq n$ then there exist matrices $A\in\mathrm{GL}_m(\F_q)$ and $B\in\mathrm{GL}_n(\F_q)$ such that $\varphi(M)=AMB$ for all $M\in\F_q^{m\times n}$.
\item If $m= n$ then there exist matrices $A,B\in\mathrm{GL}_n(\F_q)$ such that either  $\varphi(M)=AMB$ for all $M\in\F_q^{n\times n}$, or  $\varphi(M)=AM^tB$ for all $M\in\F_q^{n\times n}$.
\end{enumerate}
\end{theorem}

\begin{definition}
Two codes $\Cc_1,\Cc_2\leq\F_q^{m\times n}$ are {\bf equivalent} if there exists a linear rank-metric isometry $\varphi:\F_q^{m\times n}\rightarrow \F_q^{m\times n}$ such that $\phi(\Cc_1)=\Cc_2$.
\end{definition}

According to these definitions and Theorem \ref{classification}, we can formulate the Extension Property for rank-metric linear codes as follows.

\begin{ext}
Let $\Cc_1,\Cc_2$ be two linear codes in $\F_q^{m\times n}$. An isometry $\varphi:\Cc_1\rightarrow\Cc_2$ satisfies the Extension Property if and only if there exist two matrices $A\in\mathrm{GL}_m(\F_q)$ and $B\in\mathrm{GL}_n(\F_q)$ such that either  $\varphi(M)=AMB$ for all $M\in\Cc_1$, or  $\varphi(M)=AM^tB$ for all $M\in\Cc_1$, where the latter case can only happen if $m = n$.
\end{ext}

\section{Obstructions to the Extension Property}

In this section we discuss several obstructions to the Extension Property in the rank-metric case. A first problem arises from the fact that the transposition is an isometry of the ambient space only in the square case. This makes the composition of the transposition with the natural inclusion of $\iota:\F_q^{m\times m}\hookrightarrow\F_q^{m\times n}$, $m\leq n$, into an $\F_q$-linear isometry of $\iota(\F_q^{m\times m})\subseteq\F_q^{m\times n}$ with itself, which cannot be extended to $\F_q^{m\times n}$. This is a way of looking at the next example, due to Aleams Barra and Heide Gluesing-Luerssen.
\begin{example}[\cite{BG13}, Example 2.9]\label{example1}
Let $\Cc=\{\begin{pmatrix}
A&0
\end{pmatrix}:A\in\F_q^{2\times 2}\}\leq\F_q^{2\times 3}$ and let $\varphi:\Cc\rightarrow\Cc$ be the isometry given by $\varphi(\begin{pmatrix}
A&0
\end{pmatrix})=\begin{pmatrix}
A^t&0
\end{pmatrix}$ for all $A\in\F_q^{2\times 2}$. It is easy to see that it is not possible to extend $\varphi$ to an isometry of the whole ambient space.
\end{example}

A similar phenomenon happens in the next example, also due to Barra and Gleusing-Luerssen. 

\begin{example}[\cite{BG13}, Example 2.9]\label{example2}
Let $\Cc\leq\F_q^{4\times4}$ be the code given by
\begin{equation*}
\Cc=\left\{\begin{pmatrix}
A&0\\0&B
\end{pmatrix}:A,B\in\F_q^{2\times 2}\right\}
\end{equation*}
and consider the isometry $\varphi:\Cc\rightarrow\Cc$ given by 
\begin{equation*}
\varphi\left(\begin{pmatrix}
A&0\\0&B
\end{pmatrix}\right)=\begin{pmatrix}
A&0\\0&B^t
\end{pmatrix}
\end{equation*}
As before, one can check that $\varphi$ cannot be extended to an isometry of $\F_q^{4\times4}$.
\end{example}

In general, the natural inclusion $\iota:\F_q^{m\times m}\times\F_q^{n\times n}\hookrightarrow\F_q^{(m+n)\times(m+n)}$ is an isometry with respect to the sum-rank metric in the domain and the rank metric in the codomain. When composed with the product of the identity on $\F_q^{m\times m}$ and the transposition on $\F_q^{n\times n}$, it yields an isometry of $\iota(\F_q^{m\times m}\times\F_q^{n\times n})\subseteq\F_q^{(m+n)\times(m+n)}$ with itself, which does not extend to $\F_q^{(m+n)\times(m+n)}$.

We stress that, in both examples there is a smaller, natural ambient space to which the isometry can be extended. In fact even more, in those specific examples the isometries are already defined on a smaller ambient space (on which therefore they can be trivially extended). In the first example, the isometry is defined on $\F_q^{2\times 2}$ while in the second example it is defined on $\F_q^{2\times 2}\times \F_q^{2\times 2}$, naturally endowed with the sum-rank metric. In order to avoid such problems, one may want to consider codes that cannot be contained in a smaller ambient space, that is, such that $\rowsp(\Cc)=\F_q^n$ and $\colsp(\Cc)=\F_q^m$. 

\medskip

We now discuss a different obstruction to the Extension Property. Let $\varphi$ be an isometry of $\F_q^{m\times n}$. Then for every $\Cc\leq\F_q^{m\times n}$ we have that 
\begin{equation}\label{r&c}
\dim(\rs(\Cc))=\dim(\rs(\varphi(\Cc)))\text{ and }\dim(\cs(\Cc))=\dim(\cs(\varphi(\Cc))).   
\end{equation} Therefore, in order to be extendable, an isometry must satisfy this property. The next example shows that not all linear isometries do. 

\begin{example}\label{example3}
Let $\Cc_1,\Cc_2\in\F_2^{2\times3}$ be the codes 
\begin{equation*}
\Cc_1=\left\langle\begin{pmatrix}
1&1&0\\0&1&0
\end{pmatrix},\begin{pmatrix}
0&1&0\\1&0&0
\end{pmatrix}\right\rangle\,\,\,\,\,\,\, \Cc_2=\left\langle\begin{pmatrix}
0&0&1\\0&1&0
\end{pmatrix},\begin{pmatrix}
0&1&0\\1&0&0
\end{pmatrix}\right\rangle
\end{equation*}
and let $\varphi:\Cc_1\rightarrow \Cc_2$ be the $\F_2$-linear map given by
\begin{equation*}
\varphi\left(\begin{pmatrix}
1&1&0\\0&1&0
\end{pmatrix}\right)=\begin{pmatrix}
0&0&1\\0&1&0
\end{pmatrix}\,\,\,\,\,\,\,\varphi\left(\begin{pmatrix}
0&1&0\\1&0&0
\end{pmatrix}\right)=\begin{pmatrix}
0&1&0\\1&0&0
\end{pmatrix}\,.
\end{equation*}
Since $\Cc_1$ and $\Cc_2$ are codes of constant rank 2, then $\varphi$ is an isometry. Notice that $\dim(\rs(\Cc_1))=2$ while $\dim(\rs(\Cc_2))=3$. In particular, $\varphi$ cannot be extended to an isometry of $\F_2^{2\times3}$.
\end{example}

The last example motivates us to look at isometries $\varphi:\Cc_1\rightarrow\Cc_2\leq \F_q^{m\times n}$ with the following property, which implies (\ref{r&c}).

\begin{property}\label{propP}
There exist $A\in\mathrm{GL}_m(\F_q)$ and $B\in\mathrm{GL}_n(\F_q)$ such that
$$\rs(\varphi(C))=\rs(CB) \mbox{ and } \cs(\varphi(C))=\cs(AC)$$ for all $C\in\Cc_1$.
\end{property}

Notice that none of the isometries considered in Examples \ref{example1}, \ref{example2} and \ref{example3} satisfy Property \ref{propP}. While Property \ref{propP} is necessary for the Extension Property to hold, it is not sufficient, as the next example shows. 

\begin{example}\label{arxiv_ex}
In \cite[Example 1]{CKWW} the authors exhibit three distinct equivalence classes of MRD codes in $\F_2^{4\times 4}$ with minimum distance $4$. Any $\F_2$-linear map between codes in different equivalent classes is an isometry, since each nonzero element has rank 4. Moreover, each of these maps satisfy Property \ref{propP} with $A=B=\mathrm{Id}$. A proof that these codes do not satisfy the Extension Property appeared in the first arXiv version of the same paper as \cite[Example 7.1]{CKWWv1}.
\end{example}

The obstruction to the Extension Property in Example \ref{arxiv_ex} can be seen as coming from the interaction between the linear structure of the code and the group structure of the code without the zero matrix. More precisely, if $\Cc$ is a vector space of square matrices and $\Cc\setminus\{0\}$ is a subgroup of the general linear group, then  every $\F_q$-linear isomorphism from $\Cc$ to itself is a linear isometry. Moreover, if it fixes the identity and it has the Extension Property, then it is a group homomorphism. Therefore, any $\F_q$-linear isomorphism from $\Cc$ to itself which fixes the identity and is not a group homomorphism cannot have the Extension Property.

\begin{example}
Let $P\in\gl_n(\F_q)$ of order $q^n-1$, let $Q=P^{q-1}$. Let $\Cc=\F_q[P]=\langle P\rangle\cup\{0\}\subseteq\F_q^{n\times n}$. Every nonzero element of $\Cc$ has rank $n$, hence any injective $\F_q$-linear isomorphism of $\Cc$ with itself is an isometry. Both $P$ and $Q$ are linearly independent from the identity matrix $\Id$, so there is a linear isometry $\varphi:\Cc\rightarrow\Cc$ with $\varphi(\Id)=\Id$ and $\varphi(P)=Q$. If $\varphi$ has the Extension Property, then either $\varphi(M)=AMA^{-1}$ or $\varphi(M)=AM^tA^{-1}$ for some $A\in\gl_n(\F_q)$. Therefore $Q=\varphi(P)\in\{APA^{-1},AP^tA^{-1}\}$, however $Q$ has order $q^{n-1}+q^{n-2}+\ldots+1$, while $APA^{-1}$ and $AP^tA^{-1}$ have order $q^n-1$.
\end{example}

Even when $\Cc\setminus\{0\}$ is not a group, an isometry on a set of square matrices which fixes the identity and for which the Extension Property holds needs to be multiplicative. This constitutes an obstruction to the Extension Property, since not every linear isometry is multiplicative.

\begin{example}
Let $\Cc\in\F_2^{3\times3}$ be the code given by
\begin{equation*}
\Cc=\left\{0,\Id,
\begin{pmatrix}
1&0&0\\1&1&0\\0&0&0
\end{pmatrix},
\begin{pmatrix}
0&0&0\\1&0&0\\0&0&1
\end{pmatrix}\right\}
\end{equation*}
and let $\varphi:\Cc\rightarrow \Cc$ be the isometry of $\Cc$ with itself that fixes the identity matrix and swaps the other two matrices.

Suppose that $\varphi$ can be extended to an isometry of the whole ambient space. Then, there are $A,B\in\gl_3(\F_2)$ such that either $\varphi(C)=ACB$ for all $C\in\Cc$ or $\varphi(C)=AC^tB$ for all $C\in\Cc$. Since $\varphi(\Id)=\Id$, we have that $AB=\Id$ and so $B=A^{-1}$. Therefore, we obtain that
\begin{equation*}
\begin{split}
\varphi\left(\begin{pmatrix}
1&0&0\\0&1&0\\0&0&0
\end{pmatrix}\right)&=\varphi\left(\begin{pmatrix}
1&0&0\\1&1&0\\0&0&0
\end{pmatrix}\begin{pmatrix}
1&0&0\\1&1&0\\0&0&0
\end{pmatrix}\right)=\varphi\left(\begin{pmatrix}
1&0&0\\1&1&0\\0&0&0
\end{pmatrix}\right)\varphi\left(\begin{pmatrix}
1&0&0\\1&1&0\\0&0&0
\end{pmatrix}\right)=\\&=\begin{pmatrix}
0&0&0\\1&0&0\\0&0&1
\end{pmatrix}\begin{pmatrix}
0&0&0\\1&0&0\\0&0&1
\end{pmatrix}=\begin{pmatrix}
0&0&0\\0&0&0\\0&0&1
\end{pmatrix}\,.
\end{split}
\end{equation*} 
The map $\varphi$ sends an element of rank $2$ to an element of rank $1$, contradicting the assumption that $\varphi$ is an isometry. We conclude that $\varphi$ does not have the Extension Property. 
Notice however that $\varphi$ satisfies Property \ref{propP} with
$$A=\begin{pmatrix}
0&0&1\\1&1&1\\1&0&0
\end{pmatrix}\text{ and }B=\begin{pmatrix}
1&0&0\\1&0&1\\1&1&0
\end{pmatrix}\,.$$
\end{example}

Property \ref{propP} suggests to look at codes generated by rank-one elements. In fact, if $C$ is a rank-one element with row space $\langle u\rangle$ and column space $\langle v\rangle$, then $\varphi(C)$ is a rank-one element with row space $\langle uB\rangle$ and column space $\langle Av\rangle$. Therefore, $\varphi$ determines $Av$ and $uB$ up to a scalar multiple. This simple observation allows us to prove the next result.

\begin{proposition}\label{proposition:f2}
Let $\Cc_1,\Cc_2\leq\F_2^{m\times n}$ and let $\varphi:\Cc_1\rightarrow\Cc_2$ be an isometry which satisfies Property \ref{propP}. If $\Cc_1$ is generated by elements of rank 1, then $\varphi$ is extendable.
\end{proposition}

\begin{proof}
Since $\varphi$ has Property \ref{propP}, then $\varphi(C)$ and $ACB$ have the same row and column space for all $C\in\Cc$. Over $\F_2$ this give that $A^{-1}\varphi(C)B^{-1}=C$ for every $C\in\Cc_1$ of rank 1. If $\Cc_1$ is generated by elements of rank 1, we conclude by linearity that $A^{-1}\varphi(C)B^{-1}=C$ for all $C\in\Cc_1$.
\end{proof}

Even for $\Cc$ generated by elements of rank $1$, the Extension Property may fail if we do not require Property \ref{propP}.

\begin{example}
Let $\Cc\subseteq\F_2^{2\times 3}$ be the linear code generated by
\begin{equation*}
    C_1=\begin{pmatrix}
    1&0&0\\0&0&0
    \end{pmatrix},\;\; C_2=\begin{pmatrix}
    0&0&0\\0&1&0
    \end{pmatrix},\;\;
    C_3=\begin{pmatrix}
    0&0&1\\0&0&1
    \end{pmatrix},\;\; C_4=\begin{pmatrix}
    1&1&0\\1&1&0
    \end{pmatrix}.
    \end{equation*}
Let $\varphi:\Cc\rightarrow \Cc$ be the linear map given by $\varphi(C_i)=C_i$ for $i=1,2,3$ and $\varphi(C_4)=C_4+C_3$. One can verify that $\varphi$ is an isometry that cannot be extended to the whole ambient space, since it does not satisfy Property 1.
\end{example}

One may wonder whether the failure of the Extension Property is due to the fact that the code is small compared to the ambient space. The next example shows that this is not the case.

\begin{example}
Starting from the code $\Cc$ from the previous example, for each $n>3$ we construct a code $\Cc_n\in\F_2^{2\times n}$ given by
    \begin{equation*}
        \Cc=\left\{\begin{pmatrix}
        A&C
        \end{pmatrix}:A\in\F_2^{2\times(n-3)},\,C\in\Cc\right\}.
    \end{equation*}
Let $\varphi_n:\Cc_n\rightarrow \Cc_n$ be the linear map given by $\varphi_n\begin{pmatrix}A&0\end{pmatrix}=A$ for $A\in\F_2^{2\times(n-3)}$ and $\varphi_n\begin{pmatrix}0&C\end{pmatrix}=\varphi(C)$. Again, $\varphi_n$ is an isometry that cannot be extended to the whole ambient space. Moreover, notice that
    \begin{equation*}
        \lim_{n\to\infty}\frac{\dim(\Cc_n)}{\dim\left(\F_2^{2\times n}\right)}=\lim_{n\to\infty}\frac{2n-2}{2n}=1.
    \end{equation*}
This show that there exist non-extendable isometries defined on codes, whose dimension comes arbitrarily close to that of the ambient space.
\end{example}

We state the analogous result of Proposition \ref{proposition:f2} for arbitrary $q$ as an open question.

\begin{question}\label{problem1}
Let $\Cc_1,\Cc_2\leq\F_q^{m\times n}$ and let $\varphi:\Cc_1\rightarrow\Cc_2$ be an isometry which satisfies Property \ref{propP}. If $\Cc_1$ is generated by elements of rank 1, then the same is true for $\Cc_2$. If this is the case, does $\varphi$ have the Extension Property?
\end{question}
Our Main Theorem provides a positive answer to Question \ref{problem1}, for codes which are generated by elementary matrices.

Let $1\leq i\leq m$ and $1\leq j\leq n$. We denote by $E_{i,j}$ the matrix in $\F_q^{m\times n}$ that has $1$ in position $(i,j)$ and $0$ everywhere else. We call these matrices elementary. We now state our main result, which we will prove in Section \ref{sect:proof}.

\begin{maintheorem}\label{theorem:elementary}
Let $\Cc=\langle E_{i_1,j_1},\dots,E_{i_k,j_{k}}\rangle\leq\F_q^{m\times n}$ be a code generated by $k$ elementary matrices. Let $\varphi:\Cc\rightarrow\Cc$ be an isometry such that for all $1\leq h\leq k$ one has $\varphi(E_{i_h,j_h})=\alpha_h E_{i_h,j_h}$ for some $\alpha_h\in \F_q^*$. Then $\varphi$ satisfies the Extension Property.
\end{maintheorem}

The next example shows that the statement of the Main Theorem fails, if the code is generated by non-elementary, rank-one matrices.

\begin{example}
Let $q\neq 2$ and let $\Cc\in\F_q^{2\times 4}$ the code generated by the following elements of rank 1:
\begin{equation*}
\begin{split}
&C_1=\begin{pmatrix}
1&0&0&0\\0&0&0&0
\end{pmatrix},\;\;\; C_2=\begin{pmatrix}
0&0&0&0\\0&1&0&0
\end{pmatrix},\\
&C_3=\begin{pmatrix}
0&0&1&0\\0&0&2&0
\end{pmatrix},\;\;\; C_4=\begin{pmatrix}
0&0&0&1\\0&0&0&1
\end{pmatrix},\;\;\; C_5=\begin{pmatrix}
0&0&0&0\\1&1&1&1
\end{pmatrix}.
\end{split}
\end{equation*}
Let $\alpha\in\F_q\setminus\{0,1\}$ and let $\varphi:\Cc\rightarrow\Cc$ be the linear map given by
$\varphi(C_i)=C_i$ for $1\leq i\leq 4$ and $\varphi(C_5)=\alpha C_5$. 
One can check that $\varphi$ is an isometry and that it does not have the Extension Property. In fact, $\varphi$ does not satisfies Property \ref{propP}, since $\rs(C_5-C_2)\leq\rs(\sum_{i=1}^5C_i)$ but $\rs(\varphi(C_5-C_2))\cap\rs(\varphi(\sum_{i=1}^5C_i))=\{0\}$. Notice that, since $\varphi$ does not satisfies Property \ref{propP}, it does not yield a negative answer to Question \ref{problem1}. 
In addition, this example shows that it does not suffice in general to check Property \ref{propP} on a system of generators of the code.
\end{example}

\section{Matrix paths}

In this section we establish some preliminary result which we will use in the proof of the Main Theorem. We start by introducing the notion of path in a matrix. From here on, let $m,n\geq2$.

\begin{definition}\label{defn:path}
Let $M\in\F_q^{m\times n}$ be a matrix. A {\bf path} $\pi$ of length $k\in\mathbb{N}$ in $M$ is a finite ordered sequence of positions of nonzero entries $\left((i_1,j_1),(i_2,j_2),\dots(i_{k},j_{k})\right)$ such that two consecutive elements share either the first or the second component and $(i_h,j_h)\neq(i_{s},j_{s})$ for $h\neq s$.

A path $\pi$ of length at least $4$ is {\bf closed} if the first and the last entries share a component. The {\bf support} $\mathrm{supp}(\pi)$ of a path $\pi$ is the set of elements of $\pi$.
A path $\pi$ is {\bf simple} if no three entries of $\pi$ share a component.
\end{definition}

These definitions are borrowed from graph theory. Indeed, one can naturally associate to every $M\in\F_q^{m\times n}$ a finite graph $G_M=(V_M,E_M)$, such that  $V_M$ is the set of positions of the nonzero entries of $M$ and two vertices in $V_M$ are connected by an edge in $E_M$ if and only if the corresponding entries lay on a common line (that is, a common row or column). The notions of path and closed path from Definition \ref{defn:path} correspond to the usual definitions in graph theory. A path is simple if the subgraph of $G_M$ induced by the set of vertices in the path does not contain any clique.

We are mainly interested in closed simple paths. We begin by establishing some of their basic properties. First notice that, up to a cyclic permutation and to reversing the order, every simple path is determined by its support. Moreover, in the next lemma we see that the entries corresponding to the elements of a closed simple path are contained in a square submatrix with exactly two nonzero elements in each row and column.

\begin{lemma}\label{lemma:twoentries}
Let $M\in\F_q^{m\times n}$ be a matrix.
The entries of $M$ corresponding to the elements of a closed simple path are contained in a square submatrix with exactly two nonzero elements in each row and column.
\end{lemma}

\begin{proof}
Let $\pi=\left((i_1,j_1),(i_2,j_2),\dots(i_{k},j_{k})\right)$ be a closed path in $M$. By definition, each line of $M$ contains at most two nonzero entries whose position belongs to the support of $\pi$. Suppose by contradiction that there exists a line in $M$ which contains exactly one nonzero entry in position $(i_h,j_h)$. If $1<h<k$, then the three elements $(i_{h-1},j_{h-1}),(i_h,j_h),(i_{h+1},j_{h+1})$ have either the first or the second coordinate in common. If $h=1$, the same is true for $(i_1,j_1),$ $(i_2,j_2),(i_k,j_k)$. If $h=k$, the same holds for $(i_1,j_1),(i_{k-1},j_{k-1}), (i_k,j_k)$. In each case, $\pi$ is not simple. We conclude that the entries of $M$ corresponding to the elements of a closed simple path are contained in a square submatrix with exactly two nonzero elements in each row and column. In particular, it must be that $2m=2n$ and so $m=n$.
\end{proof}

The next proposition ensures that in every matrix with enough nonzero entries there is a closed simple path.

\begin{proposition}\label{proposition:closedpath}
Let $m,n\geq2$ and let $M\in\F_q^{m\times n}$ be a matrix with at least $m+n$ nonzero entries. Then there is a closed simple path in $M$.
\end{proposition}

\begin{proof}
We proceed by induction on $m+n$. If $m+n=4$ then $m=n=2$ and all the entries of the matrix are nonzero and so trivially we have a closed simple path.

Suppose now that $m+n>4$. If there exists a row in which there is at most one nonzero entry, then $m>2$. By Lemma \ref{lemma:twoentries} no close simple path can contain the position of that entry. Therefore, one may erase that row from $M$ and obtain a matrix of size $(m-1)\times n$ which contains the same paths as $M$. Similarly, one may erase any column of $M$ which contain a single nonzero entry without affecting the paths contained in $M$.

By eliminating all rows and columns of $M$ which contain at most one nonzero entry, we reduce to a matrix which contains at least two nonzero entries in each row and column. Notice that the operation of canceling any rows and columns of $M$ which contain at most one nonzero entry preserves the property that the matrix has at least as many nonzero entries as the sum of its number of rows and its number of columns. We can now build a closed simple path as follows. Starting from an arbitrary nonzero entry, move along the correspondent row and select another nonzero entry. Then move along the column of last nonzero entry picked and select another nonzero entry. Proceed in this way, alternating between rows and columns. At every step, we find a nonzero entry different from the last one that was picked, since we supposed that in each line we have at least two nonzero entries. Since the number of lines is finite, after $k$ steps we must choose an entry on a line where there is already one entry which was picked at a step $h$ with $1\leq h<k-1$. As soon as that happens, we choose that entry. The positions of the entries that we have picked are the support of a closed simple path in $M$.
\end{proof}

\begin{remark}
The result in Proposition \ref{proposition:closedpath} is optimal, in the sense that there are matrices in $\F_q^{m\times n}$ with $m+n-1$ nonzero entries that do not contain any closed simple path. An example is given by
\begin{equation*}
M=\begin{pmatrix}
1&1&\dots&1\\
1&0&\dots&0\\
\vdots&\vdots&\ddots&\vdots\\
1&0&\dots&0
\end{pmatrix}\in\F_q^{m\times n}.
\end{equation*} 
\end{remark} 


\begin{definition}
Let $m,n\geq2$ and $M\in\F_q^{m\times n}$. We say that a matrix $M'\in\F_q^{m\times n}$ is a {\bf path-reduction} - or just a {\bf reduction} - of $M$ if it is obtained from $M$ by changing to zero a nonzero entry that belong to a closed simple path.

A matrix $M\in\F_q^{m\times n}$ is {\bf path-irreducible} - or just {\bf irreducible} - if does not contain any closed simple path.

Let $M_1,\dots, M_{\ell}\in\F_q^{m\times n}$. We say that $(M_1,\dots,M_{\ell})$ is a {\bf path-reduction chain} if for every $1\leq i<\ell$, $M_{i+1}$ is a reduction of $M_i$ and $M_{\ell}$ is irreducible.
\end{definition}

Since in a closed simple path there are at least four entries and a matrix may have more than one closed simple path, a matrix may have several path-reductions. We illustrate the situation in the next simple example. 

\begin{example}
Consider the matrix $M\in\F_2^{3\times 5}$ given by
\begin{equation*}
M=\begin{pmatrix}
1&0&0&1&0\\
0&1&0&1&0\\
1&1&0&0&0
\end{pmatrix}\,.
\end{equation*}
The path $((1,1),(1,4),(2,4),(2,2),(3,2),(3,1))$ is closed and simple. Replacing any of the ones in $M$ yields a reduction of $M$. In particular
\begin{equation*}
M'=\begin{pmatrix}
0&0&0&1&0\\
0&1&0&1&0\\
1&1&0&0&0
\end{pmatrix}\qquad
M''=\begin{pmatrix}
1&0&0&0&0\\
0&1&0&1&0\\
1&1&0&0&0
\end{pmatrix}\,
\end{equation*}
are reductions of $M$. Notice that both $M'$ and $M''$ are irreducible.
\end{example}

The next corollary is an immediate consequence of Proposition \ref{proposition:closedpath}.

\begin{corollary}\label{corollary:irreducible}
Let $M\in\F_q^{m\times n}$. If $M$ is irreducible, than $M$ has at most $m+n-1$ nonzero entries.
\end{corollary}

Given a matrix $M\in\F_q^{m\times n}$, it is always possible to find a path-reduction chain starting from $M$. In fact, one can simply apply consecutive reductions. Since $M$ has a finite number of nonzero entries, one obtains an irreducible matrix in a finite number of steps.

\begin{proposition}\label{proposition:pathreduction}
Let $M\in\F_q^{m\times n}$. Then there exists a path-reduction chain $(M_1,\dots,M_{\ell})$ such that $M_1=M$.
\end{proposition}

Notice that one can find more than one path-reduction chain starting with the same matrix $M$. In Appendix \ref{appendix} we prove that each path-reduction chain with $M_1=M$ has the same length.

\begin{example}
Let $M\in\F_2^{3\times 3}$ be the matrix 
\begin{equation*}
M=\begin{pmatrix}
1&1&0\\
1&1&1\\
0&1&1
\end{pmatrix}\,.
\end{equation*}
Both
\begin{equation*}
\left(\begin{pmatrix}
1&1&0\\
1&1&1\\
0&1&1
\end{pmatrix},
\begin{pmatrix}
0&1&0\\
1&1&1\\
0&1&1
\end{pmatrix},
\begin{pmatrix}
0&1&0\\
1&1&1\\
0&1&0
\end{pmatrix}\right),
\end{equation*}
and
\begin{equation*}
\left(\begin{pmatrix}
1&1&0\\
1&1&1\\
0&1&1
\end{pmatrix},
\begin{pmatrix}
1&1&0\\
1&0&1\\
0&1&1
\end{pmatrix},
\begin{pmatrix}
1&1&0\\
1&0&1\\
0&1&0
\end{pmatrix}\right)
\end{equation*}
are path-reduction chains starting with $M$.
\end{example}

\section{Proof the Main Theorem}\label{sect:proof}

In order to clarify the structure of the proof of the Main Theorem, we enclose part of it in two technical lemmas. The first one shows under which conditions two maps coincide on a closed simple path.

\begin{lemma}\label{lemma:equalmaps}
Let $M\in\F_q^{m\times n}$ and let $((i_1,j_1),\dots,(i_k,j_k))$ be a closed simple path in $M$. Let $\varphi,\psi:\langle E_{i_1,j_1},\dots,E_{i_k,j_k}\rangle\rightarrow\langle E_{i_1,j_1},\dots,E_{i_k,j_k}\rangle$ two rank-preserving linear maps such that $\varphi(E_{i_h,j_h})=s_hE_{i_h,j_h}$ and $\psi(E_{i_h,j_h})=t_hE_{i_h,j_h}$, where 
$s_1,\dots,s_k,t_1,\dots,t_k\in\F_q^*$. If $s_h=t_h$ for $1\leq h<k$, then $s_k=t_k$.
\end{lemma}

\begin{proof}
For $a\in\F_q^*$, consider the matrix  $$M_a=\left(\sum_{h=1}^{k-1}E_{i_h,j_h}\right)+aE_{i_k,j_k}.$$
Since $((i_1,j_1),\dots,(i_k,j_k))$ is a closed simple path, by Lemma \ref{lemma:twoentries}, $k$ is even and the nonzero entries of $M_a$ are contained in a square submatrix of size $k/2$, whose determinant is a linear function of $a$. Hence there exists $\bar a\in\F_q^*$ such that
$\rk(M_{\bar a})=k/2-1$ and $\rk(M_a)=k/2$ for all $a\in\F_q\setminus\{\bar a\}$.

Let $M$ be the matrix given by
$$M=\left(\sum_{h=1}^{k-1}s_h^{-1}E_{i_h,j_h}\right)+\bar as_k^{-1}E_{i_k,j_k}.$$
By assumption $\rk(\psi(M))=\rk(M)=\rk(\varphi(M))=k/2-1$. Moreover, if $s_h=t_h$ for $1\leq h<k$, then
$$\psi(M)=\left(\sum_{h=1}^{k-1}E_{i_h,j_h}\right)+t_k\bar as_k^{-1}E_{i_k,j_k}\,.$$
By the uniqueness of $\bar a$ we conclude that $\bar a=t_k\bar as_k^{-1}$, hence $t_k=s_k$.
\end{proof}

The next lemma establish the Extension Property in a special case.

\begin{lemma}\label{lemma:diagonalmatrices}
Let $\varphi:\langle E_{i_1,j_1},\dots,E_{i_k,j_k}\rangle\rightarrow\langle E_{i_1,j_1},\dots,E_{i_k,j_k}\rangle\subseteq \F_q^{m\times n}$ be a rank-preserving linear map such that $\varphi(E_{i_h,j_h})=s_hE_{i_h,j_h}$, where $s_1,\dots,s_k\in\F_q$. If the matrix $M=\sum_{h=1}^kE_{i_h,j_h}$ is irreducible, then there are two diagonal invertible matrices $A\in\F_q^{m\times m}$ and $B\in\F_q^{n\times n}$ such that $$\varphi(C)=ACB$$
for all $C\in\langle E_{i_1,j_1},\dots,E_{i_k,j_k}\rangle$.
\end{lemma}

\begin{proof}
We build the matrices $A=(a_{i,j})$ and $B=(b_{i,j})$ step by step. Let $h=1$ and set $a_{i_1,i_1}=1$ and $b_{j_1,j_1}=s_1$. This guarantees that $AE_{i_1,j_1}B=s_1E_{i_1,j_1}$. At each subsequent step, choose $h\in\{1,\ldots,k\}$ among those that have not been previously chosen and such that either $a_{i_h,i_h}$ or $b_{i_h,i_h}$ has been assigned a value, if such an $h$ exists. If $a_{i_h,i_h}$ was already assigned a value, set $b_{j_h,j_h}=a_{i_h,i_h}^{-1}s_h$. If $b_{j_h,j_h}$ was already assigned a value, set $a_{i_h,i_h}=b_{j_h,j_h}^{-1}s_h$. 

Notice that at most one among $a_{i_h,i_h}$ and $b_{j_h,j_h}$ can already have an assigned value. Indeed, assume by contradiction that both $a_{i_h,i_h}$ and $b_{j_h,j_h}$ are fixed. Then there exist two simple paths $(\alpha_1,\dots,\alpha_u)$ and $(\beta_1,\dots,\beta_v)$ such that $\alpha_1=\beta_1=(i_1,j_1)$, $\alpha_u=\beta_v=(i_h,j_h)$ and $\alpha_{u-1}\neq\beta_{v-1}$. Let $z>1$ be the smallest index such that $\alpha_{z}\neq\beta_{z}$. Let $N$ be the inclusion-minimal submatrix of $M$ whose support contains $\{\alpha_{z-1},\dots,\alpha_u,\beta_{z},\dots,\beta_{v-1}\}$. Let $d,e$ be such that $N$ has size $d\times e$. Notice that $d,e\geq2$, since $\alpha_{z-1},\alpha_z$, and $\alpha_u$ are not aligned. If $\beta_{z}$ and $\alpha_{z}$ are not aligned, then every line of $N$ contains at least two nonzero entries. Otherwise, $\alpha_{z-1},\alpha_{z}$, and $\beta_z$ are aligned, then any line that does not pass through the position $\alpha_{z-1}$ contains at least two nonzero entries of $N$. Therefore, in both cases, we have $2\max\{d,e\}$ nonzero entries in a submatrix of size $d\times e$. Since $d+e\leq 2\max\{d,e\}$, by Proposition \ref{proposition:closedpath} there exists a closed simple path in $N$, contradicting the irreducibility of $M$.

If no such $h$ exists, choose any $h$ among those that have not been previously chosen and set $a_{i_h,i_h}=1$ and $b_{j_h,j_h}=s_h$. When all values of $h$ have been considered, set to $1$ all the entries on the diagonal of $A$ and $B$ which have not been assigned a value yet.
\end{proof}

\begin{remark}
The matrix $M$ in Lemma \ref{lemma:diagonalmatrices} is irreducible, which by Corollary \ref{corollary:irreducible} implies that $\dim(\langle E_{i_1,j_1},\dots,E_{i_k,j_k}\rangle)\leq m+n-1$. Notice that $m+n-1$ is the number of degree of freedom of the pair of matrices $A,B$.
\end{remark}

We conclude the section with the proof of the Main Theorem.

\begin{proof}[Proof of the Main Theorem]
If $m=1$ or $n=1$, any injective linear map is a linear isometry and the statement holds. Suppose therefore that $m,n\geq 2$ and let $M=\sum_{h=1}^{k}E_{i_h,j_h}$. By Proposition \ref{proposition:pathreduction} there exists a path-reduction chain $(M,M_2,\dots,M_{\ell})$ with $M_{\ell}$ irreducible. Consider the subset $R\subseteq\{1,\dots,k\}$ such that
$M_{\ell}=\sum_{r\in R}E_{i_r,j_r}$. By Lemma \ref{lemma:diagonalmatrices} there are two invertible matrices $A,B$ such that 
$$ AE_{i_r,j_r} B=\varphi(E_{i_r,j_r}),$$
for all $r\in R$. Following the path-reduction chain and applying $\ell-1$ times Lemma \ref{lemma:equalmaps}, we have that
$AE_{i_h,j_h}B=\varphi(E_{i_h,j_h})$, for $1\leq h\leq k$. By linearity we conclude that $\varphi(C)=ACB$ for all $C\in\Cc$.
\end{proof}

\appendix
\section{Length of path-reduction chains}\label{appendix}

In this appendix, we prove that every path-reduction chain of a matrix $M\in\F_q^{m\times n}$ has the same length.

\begin{remark}
Let $M\in\F_q^{m\times n}$ and let $\sigma_1=((i_1,j_1),\dots,(i_k,j_k))$ and $\sigma_2=((i'_1,j'_1),\dots,(i'_h,j'_h))$ be two closed simple paths. Notice that if $\mathrm{supp}(\sigma_1)\neq \mathrm{supp}(\sigma_2)$, then $\mathrm{supp}(\sigma_1)\nsubseteq \mathrm{supp}(\sigma_2)$ and vice versa.
\end{remark}

In the next lemma, we prove that if $M$ contains two distinct closed single paths, than a path-reduction chain of $M$ has length at least 3.

\begin{lemma}\label{lemma:step1}
Let $M=(m_{ij})\in\F_q^{m\times n}$, let $\sigma_1=((i_1,j_1),\dots,(i_k,j_k))$ and $\sigma_2=((i'_1,j'_1),\dots,(i'_h,j'_h))$ be two closed simple paths such that $\mathrm{supp}(\sigma_1)\neq\mathrm{supp}(\sigma_2)$. If $(i_1,j_1)=(i'_1,j'_1)$, then for each $(i_s,j_s)\in\mathrm{supp}(\sigma_1)\setminus \mathrm{supp}(\sigma_2)$ there is a closed simple path in $M-m_{i_1,j_1}E_{i_1,j_1}$ that contains $(i_s,j_s)$.
\end{lemma}

\begin{proof}
Up to reversing the order of $\sigma_2$ and to a transposition, we may suppose without loss of generality that $j_1'=j_2'=j_k=j_1$. As a consequence, also  $i_1=i_2=i_h'=i_1'$.
Consider the list of positions
$$\gamma=(\gamma_1,\dots,\gamma_{h+k-2})=((i_2,j_2),\dots,(i_k,j_k),(i'_2,j'_2),\dots,(i'_h,j'_h)).$$
Notice that $\gamma$ is not always a path, since it can contain more than two entries with the same first or second coordinate, as well as repeated entries. 
Fix an $s$ such that $(i_s,j_s)\in\mathrm{supp}(\sigma_1)\setminus \mathrm{supp}(\sigma_2)$ and let $\gamma_x=(i_s,j_s)$. We now recursively build a finite sequence of simple paths $\pi_n$, whose support is contained in that of $\gamma$ and which start with $\gamma_x$. Let $\pi_1=(\gamma_x)$. 
Suppose that we have constructed $\pi_{n-1}=(p_1,\dots,p_{\ell})$ with $p_1=\gamma_x$ and $p_\ell=\gamma_y$, with $y=x+n-2$ mod. $h+k-2$ and $\ell\geq 2$. 
Let $z=y+1$ mod. $h+k-2$ and define $\pi_n$ as follows:
\begin{itemize}
\item If no two entries of $\pi_{n-1}$ have either the first or the second coordinate in common with $\gamma_z$, then let $\pi_n=(p_1,\ldots,p_\ell,\gamma_z)$.
\item If there exists $1\leq r<t\leq\ell$ such that $p_r,p_t$ and $\gamma_z$ share either the first or the second component, then let $\pi_n=(p_1,\ldots,p_r,\gamma_z)$ if $t=r+1$. Notice that if $t\neq r+1$, then $\pi_{n-1}$ is a closed simple path. 
\end{itemize}
For $n\geq 2$, $\pi_n$ is a simple path of length at least $2$. If for some $n$ we find a closed simple path, then we are done. Else, $\pi_{h+k-2}$ is a closed simple path, since $\gamma_{x-1}$ and $\gamma_x$ lay on a common line and $\gamma_{x-2}$ and $\gamma_{x+1}$ do not.
\end{proof}

The next lemma shows that the length of a path-reduction chain is independent of the order of the reductions.

\begin{lemma}\label{lemma:step2}
Let $M\in\F_q^{m\times n}$ and let $M,M_2,\dots,M_{k+1}$ be a path-reduction chain for $M$. Let $\alpha_{1},\ldots,\alpha_k$ be the ordered list of positions of the entries that we set to zero during the path-reduction chain. Any permutation of the sequence $\alpha_{1},\ldots,\alpha_k$ still yields a path-reduction chain for $M$.
\end{lemma}

\begin{proof}
Since the group of permutation of $k$ elements is generated by the $k-1$ transpositions $(1,2), (2,3), \ldots, (k-1,k)$, it suffices to prove that setting to zero the entries in position
$$\alpha_{1},\ldots,\alpha_{i-2},\alpha_i,\alpha_{i-1},\alpha_{i+1},\ldots,\alpha_k$$ 
in the given order gives a path-reduction chain for $M$, for $i=2,\ldots,k$. This corresponds to the sequence of matrices 
$$M_1,M_2,\ldots,M_{i-1},\bar{M}_{i},M_{i+1},M_{i+2},\ldots,M_{k+1}$$
where we let $M_1=M$. 
By assumption, $M_{k+1}$ is irreducible and $M_{j}$ is a reduction of $M_{j-1}$ for $j=2,\ldots,i-1,i+2,\ldots,k$. 

The matrix $\bar{M}_{i}$ is obtained from $M_{i-1}$ by setting to zero the entry in position $\alpha_i$. Since $\alpha_i$ belongs to a closed simple path $\pi$ in $M_i$ and every nonzero entry in $M_i$ is also a nonzero entry in $M_{i-1}$, then $\pi$ is also a closed simple path in $M_{i-1}$. Therefore, $\bar{M}_{i}$ is a reduction of $M_{i-1}$. In order to prove that $M_{i+1}$ is a reduction of $\bar{M}_{i}$, we need to show that there is a closed simple path in $\bar{M}_{i}$ which contains $\alpha_{i-1}$. Notice that $\bar{M}_{i}$ is equal to $M_i$, except for the entries in position $\alpha_{i-1}$ and $\alpha_i$. By assumption, there are closed simple paths $\sigma_1$ and $\sigma_2$ such that $\sigma_1$ contains $\alpha_{i-1}$ and $\sigma_2$ contains $\alpha_{i}$, but not $\alpha_{i-1}$. If $\sigma_1$ does not contain $\alpha_i$, then it is a closed simple path in $\bar{M}_{i}$ which contains $\alpha_{i-1}$. If instead $\sigma_1$ contains $\alpha_i$, then by Lemma \ref{lemma:step1} there is a closed simple path in $M_{i}$ which contains $\alpha_{i-1}$ but not $\alpha_i$.
This gives a closed simple path in $\bar{M}_{i}$ which contains $\alpha_{i-1}$.
\end{proof}

We are now ready to prove that every path reduction chain of a given matrix has the same length.

\begin{theorem}
Let $M\in\F_q^{m\times n}$ be a matrix. Every path-reduction chain of $M$ has the same length.
\end{theorem}

\begin{proof}
We proceed by induction on the maximum length $\ell$ of a path-reduction chain of $M$. Notice that $\ell\geq 1$ and equality holds if and only if $M$ is irreducible. If $\ell=2$, then $M$ needs to have at least one closed simple path. Moreover, there is an $\alpha$ in the path such every closed simple path in $M$ contains $\alpha$. If $M$ contains two distinct closed simple paths through $\alpha$, then by Lemma \ref{lemma:step1} it also contains a closed simple path that does not pass through $\alpha$. It follows that $M$ contains exactly one closed simple path and every path-reduction chain has length two and is obtained by replacing with zero one of the entries of $M$ in one of the positions on the closed simple path.

Let $M,M_2,\dots,M_{\ell}$ and $M,M_2',\dots, M_k'$ be two path-reduction chains for $M$, $\ell\geq k$. Let $\alpha_1,\dots,\alpha_{k-1}$ and $\beta_1,\dots,\beta_{\ell}$ be the positions of the entries of $M$ that we replace with zero to obtain the path-reduction chains $M,M_2',\dots, M_k'$ and $M,M_2,\dots, M_{\ell}$, respectively. Notice that $M_2,\dots,M_{\ell}$ is a path-reduction chain for $M_2$ and, by the induction hypothesis, every path reduction chain for $M_2$ has length $\ell-1$. Starting from $M_2$, we construct a path-reduction chain $M_2,\bar M_3,\ldots,\bar M_{\ell}$ as follows. At each step $i=1,\ldots,k-1$, if there is a closed simple path that contains $\alpha_i$, we replace the entry in position $\alpha_i$ by zero. We claim that we delete at most $k-2$ entries of $M_2$. In fact, if setting to zero the entries in position $\beta_1,\alpha_1,\ldots,\alpha_{k-1}$ in the prescribed order yields a path-reduction chain of $M$, by Lemma \ref{lemma:step2} so does setting to zero the entries in position $\alpha_1,\ldots,\alpha_{k-1},\beta_1$. This contradicts the assumption that $M,M_2',\dots, M_k'$ is a path-reduction chain. So we have obtained a path-reduction chain for $M_2$ of length $\ell-1\leq k-1$. It follows that $\ell=k$.
\end{proof}

\bibliographystyle{plain}
\bibliography{extensionbib}

\end{document}